\documentclass[conference]{IEEEtran}
\IEEEoverridecommandlockouts

\usepackage[utf8]{inputenc}
\usepackage[english]{babel}

\usepackage{cite}
\usepackage{amsmath,amssymb,amsfonts}
\usepackage{algorithm}
\usepackage{algpseudocode}

\makeatletter
\algrenewcommand\ALG@beginalgorithmic{\footnotesize}
\makeatother

\usepackage{afterpage}
\usepackage{graphicx}
\usepackage{textcomp}
\usepackage{xcolor}
\usepackage{amsthm}
\usepackage{paralist}
\usepackage{mathtools, nccmath}
\usepackage{caption}
\usepackage{subcaption}
\usepackage{setspace}
\usepackage{makecell}
\makeatletter
\newcommand\fs@betterruled{%
	\def\@fs@cfont{\bfseries}\let\@fs@capt\floatc@ruled
	\def\@fs@pre{\vspace*{5pt}\hrule height.8pt depth0pt \kern2pt}%
	\def\@fs@post{\kern2pt\hrule\relax}%
	\def\@fs@mid{\kern2pt\hrule\kern2pt}%
	\let\@fs@iftopcapt\iftrue}
\floatstyle{betterruled}
\restylefloat{algorithm}
\makeatother

\newtheorem{theorem}{Theorem}
\newtheoremstyle{exampstyle}
{1pt} 
{1pt} 
{} 
{} 
{\bfseries} 
{.} 
{.5em} 
{} 
\theoremstyle{exampstyle} 
\theoremstyle{exampstyle} \newtheorem{remark}{Remark}
\theoremstyle{exampstyle} \newtheorem{definition}{Definition}
\theoremstyle{exampstyle} \newtheorem{lemma}{Lemma}
\theoremstyle{exampstyle} 

\DeclarePairedDelimiter{\ceil}{\lceil}{\rceil}

\def\BibTeX{{\rm B\kern-.05em{\sc i\kern-.025em b}\kern-.08em
		T\kern-.1667em\lower.7ex\hbox{E}\kern-.125emX}}
\begin{document}

\title{Client-Edge-Cloud Hierarchical Federated Learning}

\author{
	\IEEEauthorblockN{Lumin Liu$^\star$, Jun Zhang$^\dagger$, S.H. Song$^\star$, and Khaled B. Letaief$^{\star \ddagger}$, \emph{Fellow, IEEE}}\\
	\IEEEauthorblockA{$^\star$ Dept. of ECE, The Hong Kong University of Science and Technology, Hong Kong\\
			$^\dagger$Dept. of EIE, The Hong Kong Polytechnic University, Hong Kong\\
			$^\ddagger$Peng Cheng Laboratory, Shenzhen, China\\
		Email:{ lliubb@ust.hk,
			jun-eie.zhang@polyu.edu.hk,
			eeshsong@ust.hk,
			eekhaled@ust.hk}}
}               

\maketitle

\begin{abstract}
Federated Learning is a collaborative machine learning framework to train a deep learning model without accessing clients' private data. Previous works assume one central parameter server either at the cloud or at the edge. The cloud server can access more data but with excessive communication overhead and long latency, while the edge server enjoys more efficient communications with the clients. To combine their advantages, we propose a client-edge-cloud hierarchical Federated Learning system, supported with a HierFAVG algorithm that allows multiple edge servers to perform partial model aggregation.  In this way, the model can be trained faster and better communication-computation trade-offs can be achieved. Convergence analysis is provided for HierFAVG and the effects of key parameters are also investigated, which lead to qualitative design guidelines. Empirical experiments verify the analysis and demonstrate the benefits of this hierarchical architecture in different data distribution scenarios. Particularly, it is shown that by introducing the intermediate edge servers, the model training time and the energy consumption of the end devices can be simultaneously reduced compared to cloud-based Federated Learning.
	
\end{abstract}


\section{Introduction} \label{Introduction}
	\par
	The recent development in deep learning has revolutionalized many application domains, such as image processing, natural language processing, and video analytics \cite{deeplearningbook16goodfellow}. So far deep learning models are mainly trained at some powerful computing platforms, e.g., a cloud datacenter, with centralized collected massive datasets. Nonetheless, in many applications, data are generated and distributed at end devices, such as smartphones and sensors, and moving them to a central server for model training will violate the increasing privacy concern. Thus, privacy-preserving distribtued training has started to receive much attention.
	In 2017, Google proposed Federated Learning (FL) and a Federated Averaging (FAVG) algorithm \cite{FL17McMahan} to train a deep learning model without centralizing the data at the data center. With this algorithm, local devices download a global model from the cloud server, perform several epochs of local training, and then upload the model weights to the server for model aggregation. The process is repeated until the model reaches a desired accuracy, as illustrated in Fig. \ref{fig1}. 
	
	\setlength{\textfloatsep}{2pt plus 1.0pt minus 2.0pt}
	\begin{figure}[t] 
		\centerline{\includegraphics[keepaspectratio=true,scale=0.52]{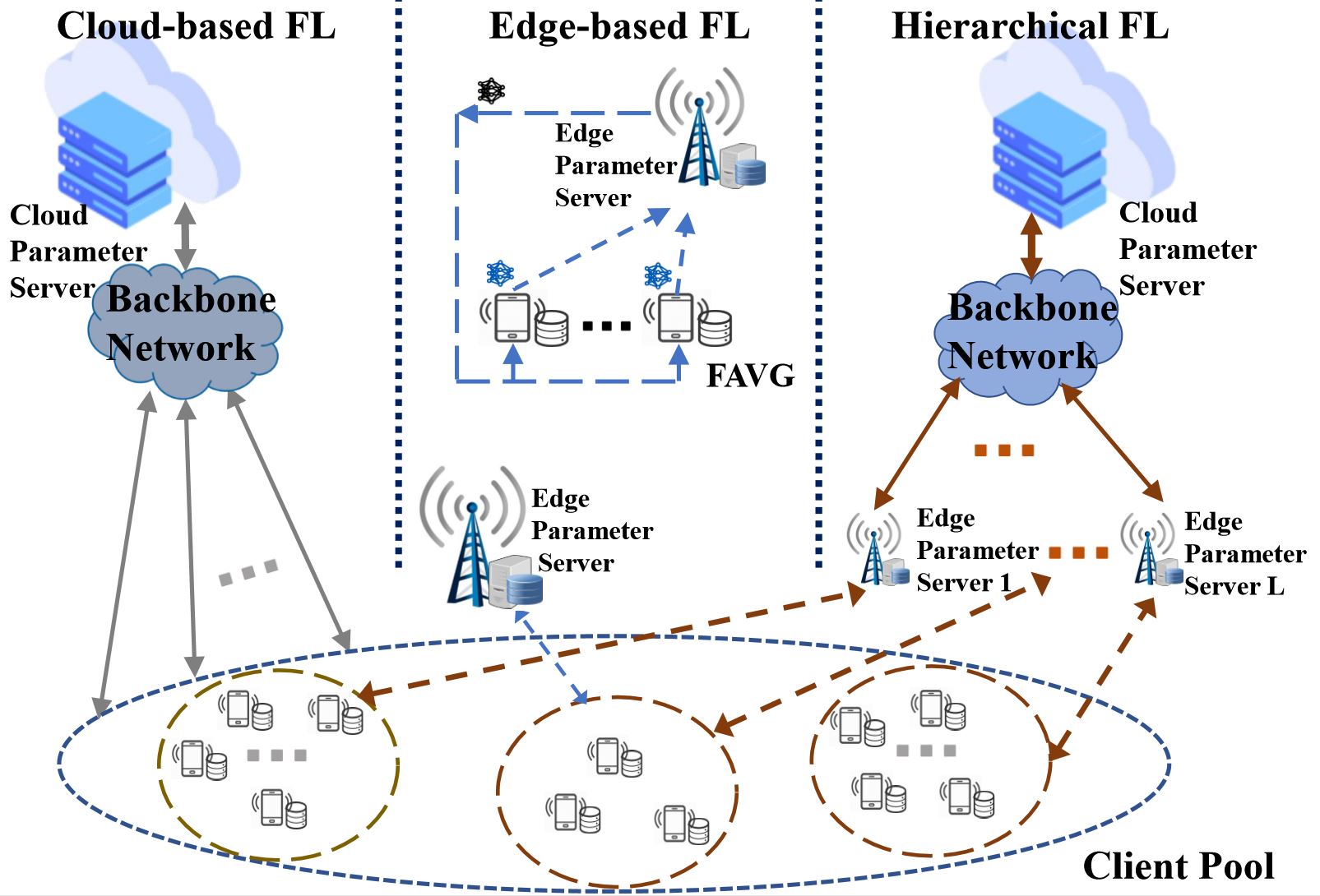}} 
		\caption{Cloud-based, edge-based and client-edge-cloud hierarchical FL. The process of the FAVG algorithm is also illustrated.}
		\label{fig1}
	\end{figure}
	\par 
	FL enables fully distributed training by decomposing the process into two steps, i.e., parallel model update based on local data at clients and global model aggregation at the server. Its feasibility has been verified in real-world implementation \cite{Keyboardprediction18Hard}. Thereafter, it has attracted great attentions from both academia and industry \cite{FLSurvey19Yang}. While most initial studies of FL assumed a cloud as the parameter server, with the recent emergence of edge computing platforms \cite{MECSurveyMao}, researchers have started investigating edge-based FL systems \cite{ClientSelection19Nishio, JSAC18Wang, FLOptimizationModel19Tran}. For edge-based FL, the proximate edge server will act as the parameter server, while the clients within its communication range of the server collaborate to train a deep learning model.
	\par
	While both the cloud-based and edge-based FL systems apply the same FAVG algorithm, there are some fundamental differences between the two systems, as shown in Fig. \ref{fig1}. In cloud-based FL, the participated clients in total can reach millions \cite{FLsys19Bonawitz}, providing massive datasets needed in deep learning. Meanwhile, the communication with the cloud server is slow and unpredictable, e.g., due to network congestion, which makes the training process inefficient \cite{FL17McMahan, FLsys19Bonawitz}. Analysis has shown a trade-off between the communication efficiency and the convergence rate for FAVG \cite{niidconvergence19Li}. Specifically, less communication is required at a price of more local computations. On the contrary, in edge-based FL, the parameter server is placed at the proximate edge, such as a base station. So the latency of the computation is comparable to that of communication to the edge parameter server. Thus, it is possible to pursue a better trade-off in computation and communication \cite{JSAC18Wang,FLOptimizationModel19Tran}. Nevertheless, one disadvantage of edge-based FL is the limited number of clients each server can access, leading to inevitable training performance loss.
	\par
	From the above comparison, we see a necessity in leveraging a cloud server to access the massive training samples, while each edge server enjoys quick model updates with its local clients. This motivates us to propose a client-edge-cloud hierarchical FL system as shown on the right side of Fig. \ref{fig1}, to get the best of both systems. Compared with cloud-based FL, hierarchical FL will significantly reduce the costly communication with the cloud, supplemented by efficient client-edge updates, thereby, resulting a significant reduction in both the runtime and number of local iterations. On the other hand, as more data can be accessed by the cloud server, hierarchical FL will outperform edge-based FL in model training. These two aspects are clearly observed from Fig. \ref{fig3}, which gives a preview of the results to be presented in this paper. While the advantages can be intuitively explained, the design of a hierarchical FL system is nontrivial. First, by extending the FAVG algorithm to the hierarchical setting, will the new algorithm still converge? Given the two levels of model aggregation (one at the edge, one at the cloud), how often should the models be aggregated at each level? Morever, by allowing frequent local updates, can a better latency-energy tradeoff be achieved? In this paper, we address these key questions. First, a rigorous proof is provided to show the convergence of the training algorithm. Through convergence analysis, some qualitavie guidelines on picking the aggregation frequencies at two levels are also given. Experimental results on \textit{MNIST} \cite{MNIST09} and \textit{CIFAR-10} \cite{CIFAR10-09} datasets support our findings and demonstrate the advantage of achieving better communication-computation tradeoff compared to cloud-based systems.
	
\section{Federated Learning Systems} \label{FLsys}

	In this section, we first introduce the general learning problem in FL. The cloud-based and edge-based FL systems differ only in the communication and the number of participated clients, and they are identical to each other in terms of architecture. Thus, we treat them as the same traditional two-layer FL system in this section and introduce the widely adopted FAVG \cite{FL17McMahan} algorithm. For the client-edge-cloud hierarchical FL system, we present the proposed three-layer FL system, and its optimization algorithm, namely, HierFAVG.

	\subsection{Learning Problem}\label{Learning Problem}
	\par
	We focus on supervised Federated Learning. Denote $\mathcal{D} = { \{ \boldsymbol{x}_j,y_j \}}_{j=1}^{|\mathcal{D}|}$ as the training dataset, and $|\mathcal{D}|$ as the total number of training samples, where $\boldsymbol{x}_j$ is the $j$-th input sample, $y_j$ is the corresponding label. $\boldsymbol{w}$ is a real vector that fully parametrizes the ML model. $f(\boldsymbol{x}_j,y_j,\boldsymbol{w})$, also denoted as $f_j(\boldsymbol{w})$ for convenience, is the loss function of the $j$-th data sample, which captures the prediction error of the model for the $j$-th data sample.
	The training process is to minimize the empirical loss $F(\boldsymbol{w})$ based on the training dataset \cite{Optimization18Bottou}:
	\useshortskip
	\begin{equation} \label{LossFunc}
	\small
	F(\boldsymbol{w}) = \frac{1}{|\mathcal{D}|} \sum_{j=1}^{|\mathcal{D}|} f(\boldsymbol{x}_j,y_j,\boldsymbol{w}) = \frac{1}{|\mathcal{D}|} \sum_{j=1}^{|\mathcal{D}|} f_j(\boldsymbol{w}).
	\end{equation}
	
	The loss funciton $F(\boldsymbol{w})$ depends on the ML model and can be convex, e.g. logistic regression, or non-convex, e.g. neural networks. The complex learning problem is usually solved by gradient descent. Denote $k$ as the index for the update step, and $\eta$ as the gradient descent step size, then the model parameters are updated as:
	\vspace{-2mm}
	\begin{equation*}
	\boldsymbol{w}(k) = \boldsymbol{w}(k-1) - \eta \nabla F(\boldsymbol{w}(k-1)).
	\end{equation*}
	\vspace{-5mm}
	\par
	In FL, the dataset is distributed on $N$ clients as $\{\mathcal{D}_i\}_{i=1}^N$, with $\cup_{i=1}^N\mathcal{D}_i=\mathcal{D}$ and these distributed datasets cannot be directly accessed by the parameter server. Thus, $F(\boldsymbol{w})$ in Eq. \eqref{LossFunc}, also called the global loss, cannot be directly computed, but can only be computed in the form of a weighted average of the local loss functions $F_i(\boldsymbol{w})$, on local datasets $\mathcal{D}_i$. Specifically, $F(\boldsymbol{w})$ and $F_i(\boldsymbol{w})$ are given by:
	\vspace{-2mm}
	\begin{equation*}
	\small
	F(\boldsymbol{w}) =\frac{\sum_{i=1}^{N}|\mathcal{D}_i| F_i(\boldsymbol{w})}{|\mathcal{D}|}, \quad
	F_i(\boldsymbol{w}) = \frac{\sum_{j\in \mathcal{D}_{i}} f_j(\boldsymbol{w})}{|\mathcal{D}_i|}.
	\end{equation*}
	\vspace{-5mm}
	\setlength{\textfloatsep}{2pt plus 1.0pt minus 2.0pt}
	\begin{figure}[t] 
		\centerline{\includegraphics[keepaspectratio=true,scale=0.25]{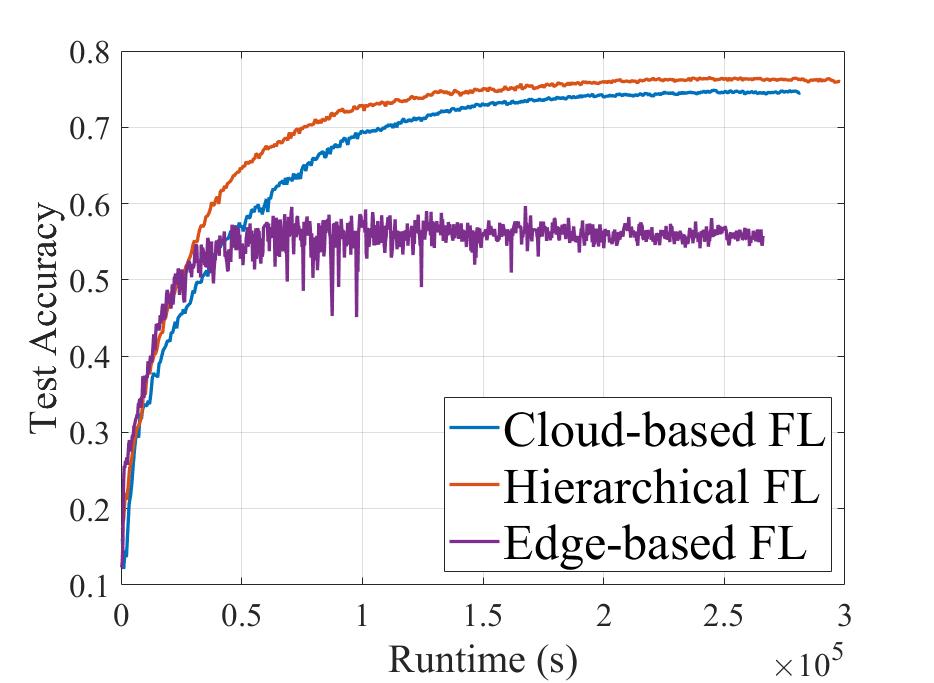}}
		\caption{Testing Accuracy w.r.t to the runtime on \textit{CIFAR}-10.}
		\label{fig3}
	\end{figure}
	\subsection{Traditional Two-Layer FL}\label{Two-Layer}
	\par
	In the traditional two-layer FL system, there are one central parameter server and $N$ clients. 
	To reduce the communication overhead, the FAVG algorithm\cite{FL17McMahan} communicates and aggreagtes after every $\kappa$ steps of gradient descent on each client. The process repeats until the model reaches a desired accuracy or the limited resources, e.g., the communication or time budget, run out.
	\par
	Denote $\boldsymbol{w}_i(k)$ as the parameters of the local model on the $i$-th client, then $\boldsymbol{w}_i(k)$ in FAVG evolves in the following way: 
	
	\begin{equation*}
	\text{\small $\boldsymbol{w}_i(k)$ } = 
	\begin{cases}
	\text{\small $\boldsymbol{w}_i(k-1) - \eta _k \nabla F_i(\boldsymbol{w}_i(k-1))$} &  \text{\footnotesize $k \mid \kappa \neq 0$}\\[8pt]
	\text{ $\frac{\sum_{i=1}^N | \mathcal{D}_i|\big[\boldsymbol{w}_i(k-1) - \eta _k \nabla F_i(\boldsymbol{w}_i(k-1)) \big]}{| \mathcal{D}|} $} 
	&
	\text{\footnotesize $k \mid \kappa = 0$}
	\end{cases} 
	\end{equation*}
	
	\subsection{Client-Edge-Cloud Hierarchical FL}\label{Hierarchical}
	In FAVG, the model aggregation step can be interpreted as a way to exchange information among the clients. Thus, aggregation at the cloud parameter server can incorporate many clients, but the communicaiton cost is high. On the other hand, aggregation at the edge parameter server only incorporates a small number of clients with much cheaper communicaiton cost. To combine their advantages, we consider a hierarchical FL system, which has one cloud server, $L$ edge servers indexed by $\ell$, with disjoint client sets $\{\mathcal{C}^\ell\}_{\ell=1}^L$, and $N$ clients indexed by $i$ and $\ell$, with distributed datasets $\{\mathcal{D}_i^\ell\}_{i=1}^N$. Denote $\mathcal{D}^\ell$ as the aggregated dataset under edge $\ell$. Each edge server aggregates models from its clients.
	\par
	With this new architecture, we extend the FAVG to a HierFAVG algorithm. The key steps of the HierFAVG algorithm proceed as follows. After every $\kappa_1$ local updates on each client, each edge server aggregates its clients' models. Then after every $\kappa_2$ edge model aggregations, the cloud server aggregates all the edge servers' models, which means that the communication with the cloud happens every $\kappa_1 \kappa_2$ local updates. The comparison between FAVG and \textit{HierFAVG} is illustrated in Fig. \ref{fig4}.
	Denote $\boldsymbol{w}_i^{\ell}(k)$ as the local model parameters after the $k$-th local update, and $K$ as the total amount of local updates performed, which is assumed to be an integer multiple of $\kappa_1 \kappa_2$.
	Then the details of the HierFAVG algorithm are presented in Algorithm \ref{alg:hfavg}. And the evolution of local model parameters $\boldsymbol{w}_i^{\ell}(k)$ is as follows:
	
	\useshortskip
	\begin{equation*}
	\text{\small $\boldsymbol{w}_i^{\ell}(k)$ } = 
	\begin{cases}
	\text{\small $\boldsymbol{w}_i^{\ell}(k-1) - \eta _t \nabla F_i^{\ell}(\boldsymbol{w}_i^{\ell}(k-1))$} &  \text{\footnotesize $k \mid \kappa_1 \neq 0$}\\[8pt]
	\text{\small $\frac{\sum_{i\in \mathcal{C}^\ell} |\mathcal{D}_i^\ell|\big[\boldsymbol{w}_i^{\ell}(k-1) - \eta _k \nabla F_i^{\ell}(\boldsymbol{w}_i^{\ell}(k-1)) \big]}{|\mathcal{D}^{\ell}|} $} 
	&
	\begin{aligned}
	\text{\footnotesize $k \mid \kappa_1 = 0$ } \\
	\text{\footnotesize $k \mid \kappa_1 \kappa_2 \neq 0$ }
	\end{aligned}
	\\[8pt]
	\text{\small $\frac{\sum_{i=1}^N |\mathcal{D}_i^\ell| \big[\boldsymbol{w}_i^{\ell}(k-1) - \eta _k\nabla F_i^{\ell}(\boldsymbol{w}_i^{\ell}(k-1)) \big] }{|\mathcal{D}|}$} 
	&
	\text{\footnotesize $k \mid \kappa_1 \kappa_2 = 0$ }	
	\end{cases} 
	\end{equation*}
	
	\setlength{\textfloatsep}{1pt}
	\begin{algorithm}[t]
		\small
		\caption{\strut Hierarchical Federated Averaging (HierFAVG)}
		\label{alg:hfavg}
		\begin{algorithmic}[1]
			\Procedure{HierarchicalFederatedAveraging}{} 
			\State {Initialized all clients with parameter $\boldsymbol{w}_0$}
			\For{$k = 1,2,\dots K$} 
			\For{each client $ i = 1,2,\dots,N$ in parallel}              
			\State {$\boldsymbol{w}_i^{\ell}(k)$ $\gets$ $\boldsymbol{w}_i^{\ell}(k-1) - \eta \nabla F_i(\boldsymbol{w}_i^{\ell}(k-1))$ }
			\EndFor
			\If{$k \mid \kappa_1 = 0 $}
			\For {each edge $\ell = 1,\dots, L$ in parallel}
			\State {$\boldsymbol{w}^{\ell}(k) \gets $ EdgeAggregation($\{\boldsymbol{w}_i^{\ell}(k) \}_{i \in \mathcal{C}^\ell}$)}
			\If{$k \mid \kappa_1 \kappa_2 \neq 0$}
			\For {each client $ i \in \mathcal{C}^\ell$ in parallel}
			\State {$\boldsymbol{w}_i^{\ell}(k)$ $\gets$ $\boldsymbol{w}^{\ell}(k)$}
			\EndFor
			\EndIf
			
			\EndFor
			
			\EndIf
			
			\If{$ k\mid \kappa_1 \kappa_2 = 0$}
			\State {$\boldsymbol{w}(k) \gets $ CloudAggregation($\{\boldsymbol{w}^{\ell}(k) \}_{\ell =1}^L$)}
			\For {each client $ i =1 \dots N$ in parallel}
			\State {$\boldsymbol{w}_i^{\ell}(k)$ $\gets$ $\boldsymbol{w}(k)$}
			\EndFor
			\EndIf
			
			\EndFor
			\EndProcedure
			
			\vspace{-3mm}
			\Statex
			\Function{EdgeAggregation}{$\ell,\{\boldsymbol{w}_i^{\ell}(k) \}_{i \in \mathcal{C}^\ell}$} //{\footnotesize \textit{Aggregate locally}}
			\State {$\boldsymbol{w}^{\ell}(k)$ $\gets$ $ \frac{\sum_{i\in \mathcal{C}^\ell}|\mathcal{D}_i^\ell| \boldsymbol{w}_i^{\ell}(k)}{|\mathcal{D}^{\ell}| } $}
			\State \Return {$\boldsymbol{w}^{\ell}(k)$}
			\EndFunction
			
			\vspace{-3mm}
			\Statex
			\Function{CloudAggregation}{$\{\boldsymbol{w}^{\ell}(k) \}_{\ell=1}^L$} //{\footnotesize \textit{Aggregate globally}}
			\State {$\boldsymbol{w}(k)$ $\gets$ $ \frac{\sum_{\ell =1}^L |\mathcal{D}^\ell | \boldsymbol{w}^{\ell}(k)}{|\mathcal{D}|} $}
			\State \Return {$\boldsymbol{w}(k)$}
			\EndFunction	
		\end{algorithmic}
	\end{algorithm}
	
	\useshortskip
	\setlength{\textfloatsep}{2pt plus 1.0pt minus 2.0pt}
	\begin{figure}[t] 
		\centerline{\includegraphics[keepaspectratio=true,scale=0.28]{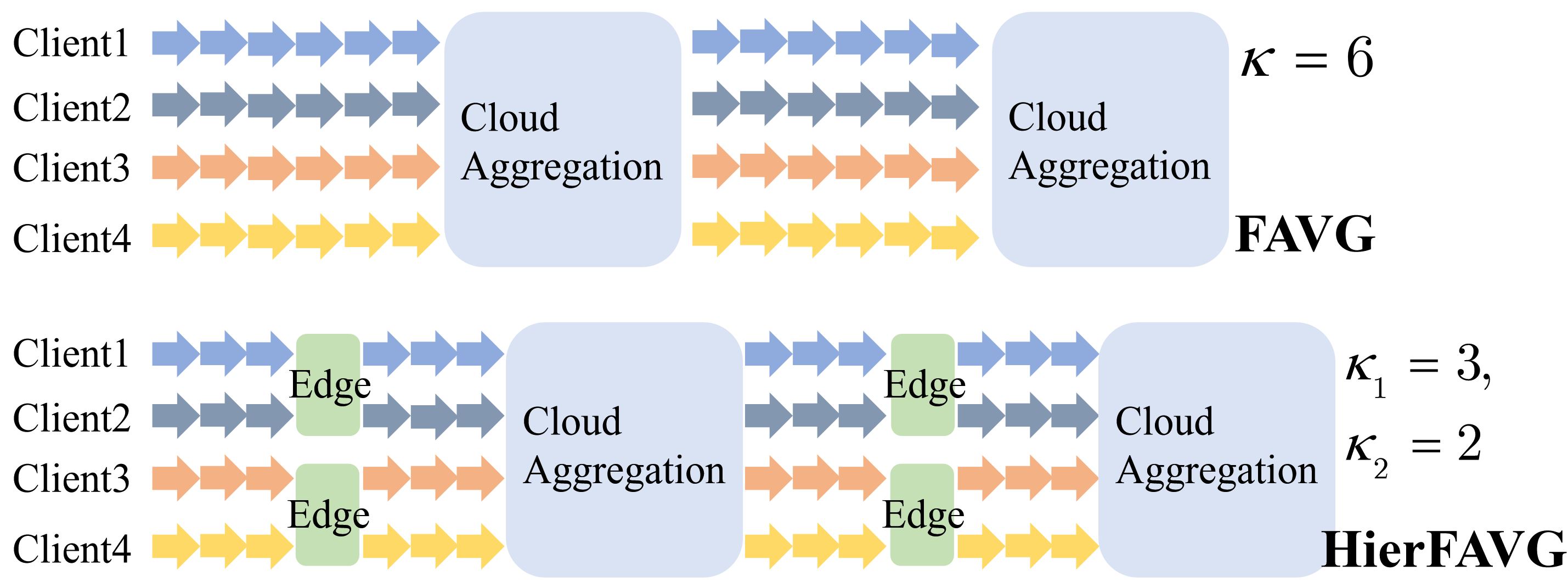}}
		\caption{Comparison of FAVG and HierFAVG.}
		\label{fig4}
	\end{figure}
\section{Convergence Analysis of HierFAVG} \label{Theoretical Analysis}
	In this section, we prove the convergence of HierFAVG for both convex and non-convex loss functions. The analysis also reveals some key properties of the algorithm, as well as the effects of key parameters.
	\useshortskip
	\subsection{Definitions}
	Some essential definitions need to be explained before the analysis. The overall $K$ local training iterations are divided into $B$ cloud intervals, each with a length of $\kappa_1 \kappa_2$, or $B\kappa_2$ edge intervals, each with a length of $\kappa_1$. The local (edge) aggregation happens at the end of each edge interval, and the global (cloud) aggregations happens at the end of each cloud interval. We use 
	$[p]$ to represent the edge interval starting from $(p-1)\kappa_1$ to $p\kappa_1$, 
	and $\{q\}$ to represent the cloud interval from $(q-1)\kappa_1 \kappa_2$ to $q\kappa_1 \kappa_2$,
	so we have $\{q\} = \cup _p [p],p = (q-1)\kappa_2+1,(q-1)\kappa_2 +2,\dots,q\kappa_2$.
	
	\begin{itemize}
	\item $F^{\ell}(\boldsymbol{w})$: The edge loss function at edge server $\ell$, is expressed as:
	\useshortskip
	\begin{equation*}
	\small
		F^{\ell}(\boldsymbol{w}) = \frac{1}{|\mathcal{D}^{\ell}|} \sum_{i\in \mathcal{C}^{\ell}}|\mathcal{D}_i^\ell | F_i(\boldsymbol{w}).
	\end{equation*}
	\vspace{-2mm}
	\item $\boldsymbol{w}(k)$: The weighted average of $\boldsymbol{w}_i^{\ell} (k)$, is expressed as:
	\useshortskip
	\begin{equation*}
	\small
	\boldsymbol{w}(k) = \frac{1}{|\mathcal{D}|}\sum_{i=1}^N |\mathcal{D}_i^\ell | \boldsymbol{\bar{w}}_i^{\ell}(k).
	\end{equation*}
	\vspace{-2mm}
	
	\item $\boldsymbol{u}_{\{q\}}(k)$: The virtually centralized gradient descent sequence, defined in cloud interval $\{q\}$, and synchronized immediately with $\boldsymbol{w}(k)$ after every cloud aggregation as:
	\vspace{-2mm}
	\begin{equation*}
	\label{uqt}
	\small
	\begin{aligned}
	&\boldsymbol{u}_{\{q\}}((q-1)\kappa_1 \kappa_2) = \boldsymbol{w} ((q-1)\kappa_1 \kappa_2),\\
	&\boldsymbol{u}_{\{q\}}(k+1) = \boldsymbol{u}_{\{q\}}(k) - \eta_k \nabla F(\boldsymbol{u}_{\{q\}}(k)).
	\end{aligned}		
	\end{equation*}
	
	\end{itemize}

	\par 
	The key idea of the proof is to show that the true weights $\boldsymbol{w}(k)$ do not deviate much from the virtuallly centralized sequence $\boldsymbol{u}_{\{q\}}(k)$. Using the same method, the convergence of two-layered FL was analyzed in \cite{JSAC18Wang}.
	
	\begin{lemma}[Convergence of FAVG\cite{JSAC18Wang}]
		For any $i$, assuming $f_i(w)$ is $\rho$-continuous, $\beta$-smooth, and convex. Als, let $F_{inf} = F(\boldsymbol{w}^*)$. If the deviation of distributed weights has an upper bound denoted as $M$, then for FAVG with a fixed step size $\eta$ and an aggregation interval $\kappa$, after $K = B\kappa$ local updates, we have the following convergence upper bound:
		\useshortskip
		\begin{equation*}
		F(w(K)) - F(\boldsymbol{w}^*) \leq \frac{1}{B(\eta \varphi - \frac{\rho M}{\kappa \varepsilon^2})}
		\end{equation*}
		when the following conditions are satisfied:
		\begin{inparaenum} 
			\small
			\item $\eta \leq \frac{1}{\beta}$
			\item $\eta \varphi - \frac{\rho M}{k \varepsilon^2} > 0$ 
			\item $F(v_{b}(b k)) - F(\boldsymbol{w}^*) \geq \varepsilon $ for $b=1,\dots,\frac{K}{\kappa}$
			\item $F(w(K)) - F(\boldsymbol{w}^*) \geq \varepsilon $
		\end{inparaenum}
		for some $\varepsilon > 0$, $\omega = \min_b \frac{1}{\|F(v_{b}((b-1)\kappa)) - F(\boldsymbol{w}^*)\|}$, $\varphi = \omega (1-\frac{\beta \eta}{2})$.
		\label{lemma3}
	\end{lemma}
	
	\vspace{1mm}
	\par
	The unique non-Independent and Identicallly Distributed (non-IID) data distribution in FL is the key property that distinguishes FL from distributed learning in datacenter. Since the data are generated seperately by each client, the local data distribution may be unbalanced and the model performance will be heavily influenced by the non-IID data distribution \cite{FLwithniid18Zhao}. In this paper, we adopt the same measurement as in \cite{JSAC18Wang} to measure the two-level non-IIDness in our hierarchical system, i.e., the client level and the edge level.
	
	\begin{definition}[Gradient Divergence]
		For any weight parameter $\boldsymbol{w}$, the gradient divergence between the local loss function of the $i$-th client, and the edge loss function of the $\ell$-th edge server is defined as an upper bound of $\|\nabla F_i^{\ell}(\boldsymbol{w}) - \nabla F^{\ell}(\boldsymbol{w}) \|$, denoted as $\delta _i^{\ell}$; the gradient divergence between the edge loss function of the $\ell_{th}$ edge server and the global loss function is defined as an upperbound of $\|\nabla F^{\ell}(\boldsymbol{w}) - \nabla F(\boldsymbol{w}) \|$, denoted as $\Delta^{\ell}$. Specifically,
		\begin{align*}
		\small
		\|\nabla F_i^{\ell}(\boldsymbol{w}) - \nabla F^{\ell}(\boldsymbol{w}) \| & \leq \delta _i^{\ell}, \\
		\|\nabla F^{\ell}(\boldsymbol{w}) - \nabla F(\boldsymbol{w}) \| & \leq \Delta^{\ell}.
		\end{align*}
		Define $\delta = \frac{\sum_{i=1}^N |\mathcal{D}_i^{\ell}|\delta_i^{\ell} }{|\mathcal{D}|} $, $\Delta = \frac{\sum_{\ell=1}^L |\mathcal{D}^{\ell}|\Delta^{\ell} }{|\mathcal{D}|} = \frac{\sum_{i=1}^N |\mathcal{D}_i^{\ell}|\Delta^{\ell} }{|\mathcal{D}|} $, and we call $\delta$ as the Client-Edge divergence, and $\Delta$ as the Edge-Cloud divergence.
	\end{definition}

	A larger gradient divergence means the dataset distribution is more non-IID. $\delta$ reflects the non-IIDness at the client level, while $\Delta$ reflects the non-IIDness at the edge level.
	
	\subsection{Convergence}
	In this section, we prove that HierFAVG converges. The basic idea is to study is how the real weights $\boldsymbol{w}(k)$ deviate from the virtually centralized sequence $\boldsymbol{u}_{\{q\}}(k)$ when the parameters in HierFAVG algorithm vary. In the following two lemmas, we prove an upper bound for the distributed weights deviation for both convex and non-convex loss functions. 
	\begin{lemma}[Convex]
		For any $i$, assuming $f_i(w)$ is $\beta$-smooth and convex, then for any cloud interval $\{q\}$ with a fixed step size $\eta_q$ and $k \in \{q\}$, we have
		\vspace{-2mm}
		\begin{equation*}
			\|\boldsymbol{w}(k) - \boldsymbol{u}_{\{q\}}(k)\| \leq G_c(k, \eta_q),
		\end{equation*}
		where 
		\vspace{-2mm}
		\begin{equation*}
		\small
			\begin{split}
			G&_c(k, \eta_q) = h(k - (q-1)\kappa_1 \kappa_2, \Delta,\eta _q)\\
			& \quad \quad \quad +h\big( k - ((q-1)\kappa_2+p(k)-1)\kappa_1,\delta, \eta _q \big) \\
			& \quad \quad \quad + \frac{\kappa_1}{2} \big(p^2(k) + p(k) -2 \big) h(\kappa_1,\delta, \eta _q),\\
			h&(x, \delta, \eta) = \frac{\delta}{\beta}\big((\eta \beta +1)^x -1\big) -\eta \beta x,\\
			p&(x) = \ceil{\frac{x}{\kappa_1} - (q-1)\kappa_2}.
			\end{split}
			\label{gap}
		\end{equation*}
		\label{lemma1}
	\end{lemma}
	\vspace{-2mm}
	\begin{remark}
		Note that when $\kappa_2 = 1$, HierFAVG retrogrades to the FAVG algorihtm. In this case, $[p]$ is the same as $\{q\}$, $p(k) =1$, $\kappa_1 \kappa_2 = \kappa_1$, and $G_c(k) = h(k - (q-1)\kappa_1, \Delta+\delta,\eta _q)$. This is consistent with the result in \cite{JSAC18Wang}. When $\kappa_1 = \kappa_2 = 1$, HierFAVG retrogrades to the traditional gradient descent. In this case, $G_c(\kappa_1 \kappa_2) = 0$, implying the distibuted weights iteration is the same as the centralized weights iteration. 
	\end{remark} 
	\begin{remark}
		\label{GradientDivergence}
		
		The following upperbound of the weights deviation, $G_c(k)$, increases as we increase either of the two aggregation intervals, $\kappa_1$ and $\kappa_2$:
		\begin{equation}
		\label{delta}
		\small
		\begin{aligned}
		G_c (k, \eta_q)&\leq G_c(\kappa_1 \kappa_2, \eta_q)\\
		&= h(\kappa_1 \kappa_2, \Delta,\eta _q) + \frac{1}{2}(\kappa_2 ^2 +\kappa_2 -1)(\kappa_1+1) h(\kappa_1, \delta,\eta _q)
		\end{aligned}
		\end{equation}
		It is obvious that when $\delta = \Delta = 0$ (i.e., which means the client data distribution is IID), we have $G_c(k) = 0$, where the distributed weights iteration is the same as the centralized weights iteration.
		\par
		When the client data are non-IID, there are two parts in the expression of the weights deviation upper bound, $G(\kappa_1 \kappa_2, \eta_q)$. The first one is caused by Edge-Cloud divergence, and is exponential with both $\kappa_1$ and $\kappa_2$. The second one is caused by the Client-Edge Divergence, which is only exponential with $\kappa_1$, but quadratic with $\kappa_2$. From lemma \ref{lemma3}, we can see that a smaller model weight deviation leads to faster convergence. This gives us some qualitative guidelines in selecting the paramters in HierFAVG:
		\begin{enumerate}
			\item When the product of $\kappa_1$ and $\kappa_2$ is fixed, which means the number of local updates between two cloud aggreagtions is fixed, a smaller $\kappa_1$ with a larger $\kappa_2$ will result in a smaller deviation $G_c(\kappa_1, \kappa_2)$. This is consistent with our intuition, namely, frequent local model averaging can reduce the number of local iterations needed. 
			\item When the edge dataset is IID, meaning $\Delta = 0$, the first part in Eq. \eqref{delta} becomes $0$. The second part is dominated by $\kappa_1$, which suggests that when the distribution of edge dataset approaches IID, increasing $\kappa_2$ will not push up the deviation upper bound much. This suggests one way to to further reduce the communication with the cloud is to make the edge dataset IID distributed.
		\end{enumerate}  
	\end{remark}
	The result for the non-convex loss function is stated in the following lemma.
	
	\begin{lemma}[Non-convex]
		For any $i$, assuming $f_i(w)$ is $\beta$-smooth, for any cloud interval $\{q\}$ with step size $\eta_q$, we have
		\vspace{-2mm}
		\begin{equation*}
		\|\boldsymbol{w}(k) - \boldsymbol{u}_{\{q\}}(k)\| \leq G_{nc}(\kappa_1 \kappa_2, \eta_q)
		\end{equation*}
		where 
		\vspace{-2mm}
		\begin{equation*}
		\small
		\begin{split}
		G_{nc}(\kappa_1 \kappa_2, \eta_q) =& h(\kappa_1 \kappa_2, \Delta,\eta _q)\\
		&+ \kappa_1 \kappa_2\frac{(1+\eta_q \beta)^{\kappa_1 \kappa_2}-1}{(1+\eta_q \beta)^{\kappa_1} -1} h(\kappa_1, \delta, \eta_q) \\
		&+ h(\kappa_1,\delta, \eta _q),\\
		\end{split}
		\end{equation*}
		\begin{equation*}
		\small
		\begin{aligned}
		h(x, \delta, \eta) &= \frac{\delta}{\beta}\big((\eta \beta +1)^x -1\big) -\eta \beta x.
		\end{aligned}
		\end{equation*}
		\label{lemma2}
	\end{lemma}
	\vspace{-3mm}
	With the help of the weight deviation upperbound, we are now ready to prove the convergence of HierFAVG for both convex and non-convex loss functions.

	\vspace{-2mm}
	\begin{theorem}[Convex]
	\label{cf}
	 For any $i$, assuming $f_i(w)$ is $\rho$-continuous, $\beta$-smooth and convex, and denoting $F_{inf} = F(\boldsymbol{w}^*)$, then after $K$ local updates, we have the following convergence upper bound of $\boldsymbol{w}(k)$ in HierFAVG with a fixed step size:
	\vspace{-2mm}
	\begin{equation*}
		\small
		F(w(K)) - F(\boldsymbol{w}^*) \leq \frac{1}{T(\eta \varphi - \frac{\rho G_c(\kappa_1 \kappa_2, \eta)}{\kappa_1 \kappa_2  \varepsilon^2})}
	\end{equation*}
	\end{theorem}
	
	\begin{proof}
		By directly substituting $M$ in Lemma \ref{lemma3} with $G(\kappa_1 \kappa_2)$ in Lemma \ref{lemma1}, we prove Theorem \ref{cf}.
	\end{proof}
	\vspace{-2mm}
	\begin{remark}
		Notice in the condition $\varepsilon ^ 2 > \frac{\rho G(\kappa_1 \kappa_2)}{\kappa_1 \kappa_2 \eta \varphi} $ of Lemma \ref{lemma3}, $\varepsilon$ does not decrease when $K$ increases. We cannot have $F(w(K)) - F(\boldsymbol{w}^*) \rightarrow 0$ as $K \rightarrow \infty$. This is because the variance in the gradients introduced by non-IIDness cannnot be eliminated by fixed-stepsize gradient descent.
	\end{remark}

	\begin{remark}
		\label{cd}
		With diminishing step sizes $\{ \eta _q \}$ that satisfy $\sum_{q=1}^{\infty} \eta_q = \infty, \sum_{q=1}^{\infty} \eta_q ^2 < \infty $, 
		the convergence upper bound for HierFAVG after $K = B\kappa_1\kappa_2$ local updates is:
		\begin{equation*}
		\small
		\label{upperbound}
		F(w(K)) - F(\boldsymbol{w}^*) \leq \frac{1}{\sum_{q=1}^B(\eta_q \varphi_q - \frac{\rho G_c(\kappa_1 \kappa_2, \eta_q)}{\kappa_1 \kappa_2 \varepsilon_q^2}) } \xrightarrow{B \rightarrow \infty}  0.
		\end{equation*}
	\end{remark}
	
	Now we consider non-convex loss functions, which appear in ML models such as neural networks.
	\vspace{-1mm}
	\begin{theorem}[Non-convex]
		For any $i$, assume that $f_i(w)$ is $\rho$-continuous, and $\beta$-smooth. Also assume that HierFAVG is initialized from $\boldsymbol{w}_0$, $F_{inf} = F(\boldsymbol{w}^*)$, $\eta_q$ in one cloud interval $\{q\}$ is a constant, then after $K =B\kappa_1\kappa_2$ local updates, the expected average-squared gradients of $F(\boldsymbol{w})$ is upper bounded as:
		\begin{equation}
			\small
			\label{nonconvex}
			\begin{split}
			\frac{\sum_{k=1}^{K} \eta_q \| \nabla F(w(k)) \|^2}{\sum_{k=1}^{K} \eta_q}  & \leq \frac{4[F(\boldsymbol{w}_0) - F(\boldsymbol{w^*})]}{\sum_{k=1}^{K} \eta_q}\\
			&+ \frac{4\rho \sum_{q=1}^{B}G_{nc}(\kappa_1, \kappa_2, \eta_q)}{\sum_{k=1}^{K} \eta_q}\\
			&+ \frac{2\beta ^2 \sum_{q=1}^{B}\kappa_1 \kappa_2 \|G_{nc}(\kappa_1 \kappa_2, \eta_q) \|^2}{\sum_{k=1}^{ K} \eta_q}.
			\end{split}
		\end{equation}
	\end{theorem}
	\vspace{-2mm}
	\begin{remark}
		When the stepsize $\{\eta_q \}$ is fixed, the weighted average norm of the gradients converges to some non-zero number.
		When the stepsize $\{\eta_q \}$ satisfies $\sum_{q=1}^{\infty} \eta_q = \infty, \sum_{q=1}^{\infty} \eta_q ^2 < \infty $, \eqref{nonconvex} converges to zero as $K \rightarrow \infty$.
	\end{remark}

\section{Experiments} \label{Experiments}
	In this section, we present simulation results for HierFAVG to verify the obeservations from the convergence analysis and illustrate the advantages of the hierarchical FL system. As shown in Fig. \ref{fig3}, the advantage over the edge-based FL system in terms of the model accuracy is obvious. Hence, we shall focus on the comparison with the cloud-based FL system.
	\vspace{-1mm}
	\subsection{Settings}
	\par
	We consider a hierarchical FL system with $50$ clients, $5$ edge servers and a cloud server, assuming each edge server authorizes the same number of clients with the same amount of training data.
	For the ML tasks, image classification tasks are considered and standard datasets \textit{MNIST} and \textit{CIFAR-10} are used.
	For the 10-class hand-written digit classification dataset \textit{MNIST}, we use the Convolutional Neural Network (CNN) with 21840 trainable parameters as in \cite{FL17McMahan}. For the local computation of the training with \textit{MNIST} on each client, we employ mini-batch Stochastic Gradient Descent (SGD) with batch size 20, and an initial learning rate 0.01 which decays exponetially at a rate of 0.995 with every epoch. For the \textit{CIFAR-10} dataset, we use a CNN with 3 convolutional blocks, which has 5852170 parameters and achieves 90\% testing accuracy in centralized training. For the local computation of the training with \textit{CIFAR-10}, mini-batch SGD is also employed with a batch size of 20, an inital learing rate of 0.1 and an exponetial learning rate decay of 0.992 every epoch. In the experiments, we also notice that using SGD with momentum can speed up training and improve the final accuracy evidently. But the benefits of the hierarchical FL system always continue to exist with or without the momentum. To be consistent with the analysis, we do not use momentum in the experiments. 
	\par
	Non-IID distribution in the client data is a key influential factor in FL. In our proposed hierarchical FL system, there are two levels of non-IIDness. In addition to the most commonly used non-IID data partition \cite{FL17McMahan}, referred to as simple NIID where each client owns samples of two classes and the clients are randomly assigned to each edge server, we will also consider the follwoing two non-IID cases for \textit{MNIST}: 
	\begin{enumerate}
		\item Edge-IID: Assign each client samples of one class, and assign each edge 10 clients with different classes. The datasets among edges are IID.
		\item Edge-NIID: Assign each client samples of one class, and assign each edge 10 clients with a total of $5$ classes of labels. The datasets among edges are non-IID.
 	\end{enumerate}	
 	\par
 	In the following, we provide the models for wireless communications and local computations \cite{FLOptimizationModel19Tran}. We ignore the possible heterogeneous communication conditions and computing resources for different clients. For the communication channel between the client and edge server, clients upload the model through a wireless channel of 1 MHz bandwidth with a channel gain $g$ equals to $10^{-8}$. The transmitter power $p$ is fixed at 0.5W, and the noise power $\sigma$ is $10^{-10}$W. For the local computation model, the number of CPU cycles to excute one sample $c$ is assumed to be 20 cycles/bit, CPU cycle frequency $f$ is 1 GHz and the effective capacitance is $2 \times 10^{-28}$. For the communication latency to the cloud, we assume it is 10 times larger than that to the edge. Assume the uploaded model size is $M$ bits, and one local iteration involes $D$ bits of data. In this case, the latency and energy consumption for one model upload and one local iteration can be caculated with the following equations (Specific paramters are shown in table \ref{tab:table2}):
 	\vspace{-2mm}
 	\begin{equation}
 	\small
 	T^{comp} = \frac{cD}{f}, \quad E^{comp} = \frac{\alpha}{2}cDf^2,
 	\end{equation}
 	\begin{equation}
 	\small
 	T^{comm} = \frac{M}{B \log _2(1+\frac{hp}{\sigma})}, \quad E^{comm} = pT^{comm}
 	\end{equation}
 	\vspace{-3mm}
 	\par
 	To investigate the local energy consumption and training time in an FL system, we define the following two metrics:
 	\begin{enumerate}
 		\item $T_\alpha$: The training time to reach a test accuracy level $\alpha$;
 		\item $E_\alpha$: The local energy consumption to reach a test accuracy level $\alpha$.
 	\end{enumerate}
 	\vspace{-2mm}
	\subsection{Results}
	\par
	We first verify the two qualitative guidelines on the key parameters in HierFAVG from the convergence analysis, i.e., $\kappa_1, \kappa_2$. The experiments are done with the \textit{MNIST} dataset under two non-IID scenarios, edge-IID and edge-NIID.
	\par 
	The first conclusion to verify is that more frequent communication with the edge (i.e., fewer local updates $\kappa_1$) can speed up the training process when the communciation frequency with the cloud is fixed (i.e., $\kappa_1\kappa_2$ is fixed.). In Fig. \ref{edgeiid1} and Fig. \ref{edgeniid1}, we fix the communication frequency with the cloud server at 60 local iterations, i.e., $\kappa_1$$\kappa_2$=60 and change the value of $\kappa_1$. For both kinds of non-IID data distribution, as we decrease $\kappa_1$, the desired accuracy can be reached with fewer training epochs, which means fewer local computations are needed on the devices. 
	\par
	The second conclusion to verify is that when the datasets among edges are IID and the communication freqency with the edge server is fixed, decreasing the communication frequency with the cloud server will not slow down the training process. In Fig. \ref{edgeiid1}, the test accuracy curves with the same $\kappa_1=60$ and different $\kappa_2$ almost coincide with each other. But for edge-NIID in Fig. \ref{edgeniid1}, when $\kappa_1=60$, increasing $\kappa_2$ will slow down the training process, which strongly supports our analysis. This property indicates that we may be able to further reduce the high-cost communication with the cloud under the edge-IID scenario, with litttle performance loss.
	\par	
	\begin{table}[t]
		\small
		\begin{center}
			\caption{The latency and energy consumption paramters for the communication and computation of \textit{MNIST} and \textit{CIFAR-10}.}
			\label{tab:table2}
			\begin{tabular}{l|c|c|c|c} 
				\hline
				\textbf{Dataset} & \textbf{$T^{comp}$} & \textbf{$T^{comm}$} & \textbf{$E^{comp}$}  & \textbf{$E^{comm}$}\\
				\hline
				\textit{MNIST} & 0.024s & 0.1233s & 0.0024J & 0.0616J\\
				\textit{CIFAR-10} & 4 & 33s & 0.4J & 16.5J\\
				\hline
			\end{tabular}
		\end{center}
	\end{table}
	
	\begin{figure}[t!]
		\label{figall}
		\begin{subfigure}[t]{0.24\textwidth}
			\centering
			\includegraphics[keepaspectratio=true,scale=0.13]{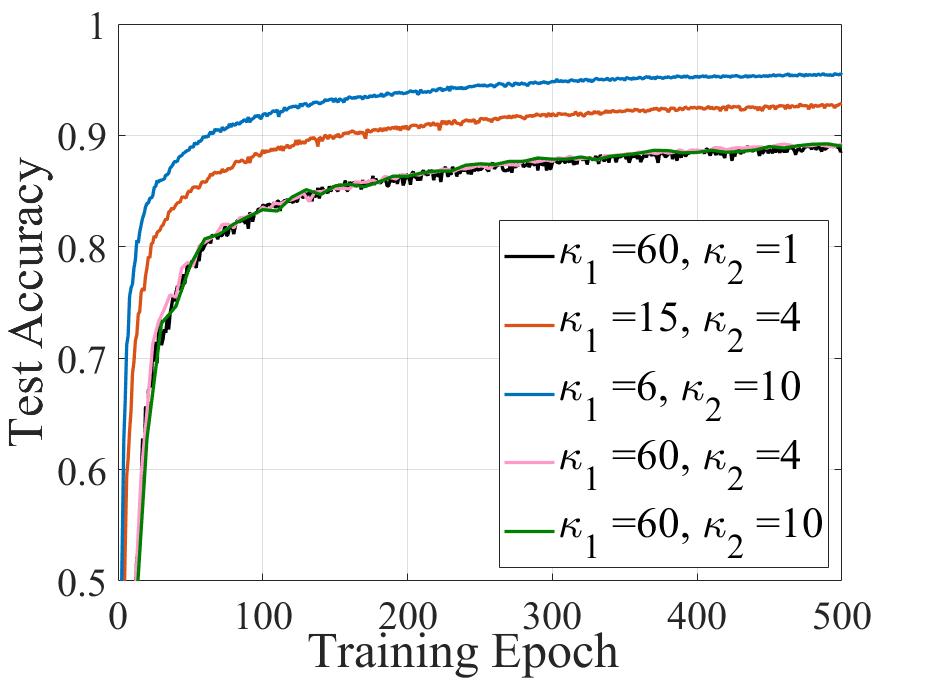}
			\caption{Edge-IID.}
			\label{edgeiid1}
		\end{subfigure}%
		\begin{subfigure}[t]{0.24\textwidth}
			\centering
			\includegraphics[keepaspectratio=true,scale=0.13]{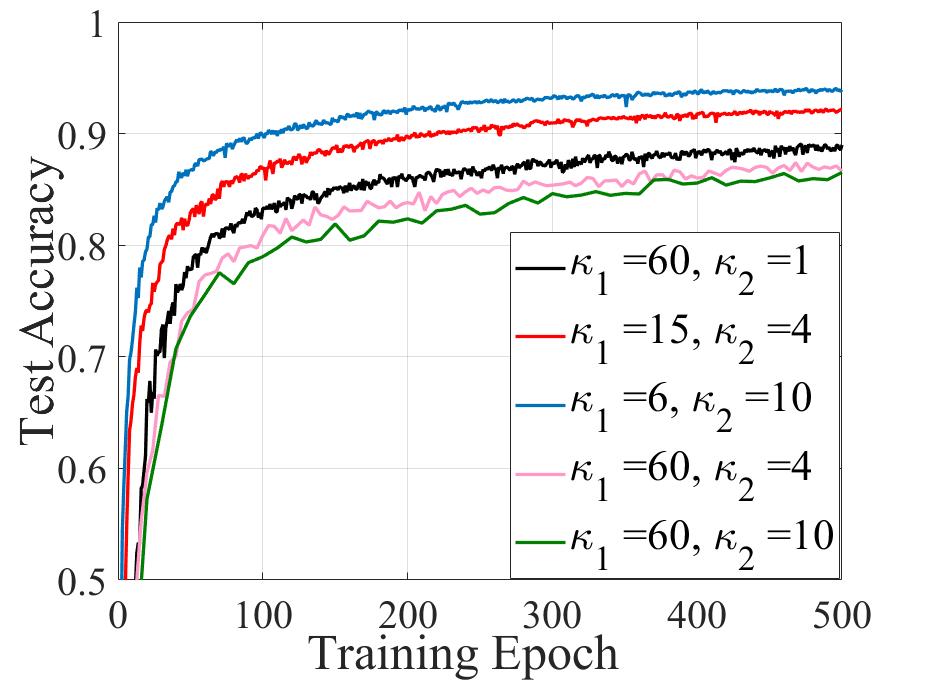}
			\caption{Edge-NIID.}
			\label{edgeniid1}
		\end{subfigure}
		\caption{Test accuracy of \textit{MNIST} dataset w.r.t training epoch.}
	\end{figure}

	\par
	Next, we investigate two critical quantities in collaborative training systems, namely, the training time and energy consumption of mobile devices. We compare cloud-based FL ($\kappa_2$=1) and hierarchical FL in Table \ref{tableall}, assuming fixed $\kappa_1\kappa_2$. A close observation of the table shows that the training time to reach a certain test accuracy decreases monotonically as we increase the communication frequency (i.e., $\kappa_2$) with the edge server for both the \textit{MNIST} and \textit{CIFAR-10} datasets. This demonstrates the great advantage in training time of the hierarchical FL over cloud-based FL in training an FL model.  For the local energy consumption, it decreases first and then increases as $\kappa_2$ increases. Because increasing client-edge communication frequency moderately can reduce the consumed energy as fewer local computations are needed. But too frequent edge-client communication also consumes extra energy for data transmission. If the target is to minimize the device energy consumption, we should carefully balance the computation and communication energy by adjusting $\kappa_1$, $\kappa_2$.
	
	\setlength{\textfloatsep}{2pt plus 1.0pt minus 2.0pt}
	\begin{table}[t!]
		\small
		\centering
		\caption{Training time and local energy consumption.}
		\begin{subtable}{.5\textwidth}
			\begin{tabular}{|c| c| c| c| c|} 
				\hline
				& \multicolumn{2}{c|}{Edge-IID}& \multicolumn{2}{c|}{Edge-NIID}\\ [0.5ex]
				\hline
				& \text{\small $E_{0.85}$(J)} & \text{\small $T_{0.85}$(s)} & \text{\small $E_{0.85}$(J)} & \text{\small $T_{0.85}(s)$}\\ 
				\hline 
				\text{\small $\kappa_1=60,\kappa_2=1$ }  & 29.4  & 385.9  & 30.8 & 405.5 \\ 
				\hline
				\text{\small $\kappa_1=30,\kappa_2=2$ } & 21.9 & 251.1 & 28.6 & 312.4 \\
				\hline
				\text{\small $\kappa_1=15,\kappa_2=4$ } & \textbf{10.1} & 177.3 & \textbf{26.9} &  218.5 \\
				\hline
				\text{\small $\kappa_1=6,\kappa_2=10$ } & 19 & \textbf{97.7} & 28.9 &  \textbf{148.4} \\
				\hline
			\end{tabular}
			\caption{\textit{MNIST} with edge-IID and edge-NIID distribution.}
			\label{table1}
		\end{subtable}
	\\
		\begin{subtable}{.5\textwidth}
			\centering
			\begin{tabular}{|c| c| c|} 
				\hline
				& \text{\small $E_{0.70}(J)$} & \text{\small $T_{0.70}(s)$}\\ 
				\hline 
				\text{\small $\kappa_1 = 50,\kappa_2=1$ }  & 7117.5 &  109800 \\ 
				\hline
				\text{\small $\kappa_1=25,\kappa_2=2$ } & \textbf{6731} & 75760 \\
				\hline
				\text{\small $\kappa_1 = 10,\kappa_2=5$ }& 9635 &  65330 \\
				\hline
				\text{\small $\kappa_1 = 5, \kappa_2=10$ }& 13135 &  \textbf{49350} \\
				\hline
			\end{tabular}
			\caption{\textit{CIFAR-10} with simple NIID distribution.}
			\label{table2}
		\end{subtable}
	\label{tableall}
	\end{table}

\section{Conclusions} \label{Conclusion}
	In this paper, we proposed a client-edge-cloud hierarchical Federated Learning architecture, supported by a collaborative training algorithm, HierFAVG. The convergence analysis of HierFAVG was provided, leading to some qualitative design guidelines. In experiments, it was also shown that it can simultaneously reduce the model training time and the energy consumption of the end devices compared to traditional cloud-based FL. While our study revealed trade-offs in selecting the values of key parameters in the HierFAVG algorithm, future investigation will be needed to fully characterize and optimize these critical parameters.

\bibliographystyle{IEEEtraN}  
\bibliography{HierFLRef.bib}  

\begin{thebibliography}{10}
\providecommand{\url}[1]{#1}
\csname url@samestyle\endcsname
\providecommand{\newblock}{\relax}
\providecommand{\bibinfo}[2]{#2}
\providecommand{\BIBentrySTDinterwordspacing}{\spaceskip=0pt\relax}
\providecommand{\BIBentryALTinterwordstretchfactor}{4}
\providecommand{\BIBentryALTinterwordspacing}{\spaceskip=\fontdimen2\font plus
\BIBentryALTinterwordstretchfactor\fontdimen3\font minus
  \fontdimen4\font\relax}
\providecommand{\BIBforeignlanguage}[2]{{%
\expandafter\ifx\csname l@#1\endcsname\relax
\typeout{** WARNING: IEEEtran.bst: No hyphenation pattern has been}%
\typeout{** loaded for the language `#1'. Using the pattern for}%
\typeout{** the default language instead.}%
\else
\language=\csname l@#1\endcsname
\fi
#2}}
\providecommand{\BIBdecl}{\relax}
\BIBdecl

\bibitem{deeplearningbook16goodfellow}
I.~Goodfellow, Y.~Bengio, and A.~Courville, \emph{Deep learning}, 2016.

\bibitem{FL17McMahan}
H.~B. McMahan, E.~Moore, D.~Ramage, and S.~Hampson, ``Communication-efficient
  learning of deep networks from decentralized data,'' \emph{Artificial
  Intelligence and Statistics}, pp. 1273--1282, April. 2017.

\bibitem{Keyboardprediction18Hard}
A.~Hard, K.~Rao, R.~Mathews, F.~Beaufays, S.~Augenstein, H.~Eichner, C.~Kiddon,
  and D.~Ramage, ``Federated learning for mobile keyboard prediction,''
  \emph{arXiv preprint arXiv:1811.03604}, 2018.

\bibitem{FLSurvey19Yang}
Q.~Yang, Y.~Liu, T.~Chen, and Y.~Tong, ``Federated machine learning: Concept
  and applications,'' \emph{ACM Trans. Intell. Syst. and Technol. (TIST)},
  vol.~10, no.~2, p.~12, 2019.

\bibitem{MECSurveyMao}
Y.~Mao, C.~You, J.~Zhang, K.~Huang, and K.~B. Letaief, ``A survey on mobile
  edge computing: The communication perspective,'' \emph{IEEE Commun. Surveys
  Tuts.}, vol.~19, no.~4, pp. 2322--2358.

\bibitem{ClientSelection19Nishio}
T.~Nishio and R.~Yonetani, ``Client selection for federated learning with
  heterogeneous resources in mobile edge,'' \emph{IEEE ICC}, May. 2019.

\bibitem{JSAC18Wang}
S.~{Wang}, T.~{Tuor}, T.~{Salonidis}, K.~K. {Leung}, C.~{Makaya}, T.~{He}, and
  K.~{Chan}, ``Adaptive federated learning in resource constrained edge
  computing systems,'' \emph{IEEE J. Sel. Areas Commun.}, vol.~37, no.~6, pp.
  1205--1221, June 2019.

\bibitem{FLOptimizationModel19Tran}
N.~H. {Tran}, W.~{Bao}, A.~{Zomaya}, N.~{Minh N.H.}, and C.~S. {Hong},
  ``Federated learning over wireless networks: Optimization model design and
  analysis,'' in \emph{IEEE INFOCOM 2019}, April 2019, pp. 1387--1395.

\bibitem{FLsys19Bonawitz}
K.~Bonawitz, H.~Eichner, W.~Grieskamp, D.~Huba, A.~Ingerman, V.~Ivanov,
  C.~Kiddon, J.~Konecny, S.~Mazzocchi, H.~B. McMahan \emph{et~al.}, ``Towards
  federated learning at scale: System design,'' \emph{Proc. of the 2nd SysML
  Conference, Palo Alto, CA, USA}, 2019.

\bibitem{niidconvergence19Li}
X.~Li, K.~Huang, W.~Yang, S.~Wang, and Z.~Zhang, ``On the convergence of fedavg
  on non-iid data,'' \emph{arXiv preprint arXiv:1907.02189}, 2019.

\bibitem{MNIST09}
A.~Krizhevsky and G.~Hinton, ``Learning multiple layers of features from tiny
  images,'' University of Toronto, Tech. Rep., 2009.

\bibitem{CIFAR10-09}
A.~Krizhevsky \emph{et~al.}, ``Learning multiple layers of features from tiny
  images,'' Citeseer, Tech. Rep., 2009.

\bibitem{Optimization18Bottou}
L.~Bottou, F.~E. Curtis, and J.~Nocedal, ``Optimization methods for large-scale
  machine learning,'' \emph{SIAM Review}, vol.~60, no.~2, pp. 223--311, 2018.

\bibitem{FLwithniid18Zhao}
Y.~Zhao, M.~Li, L.~Lai, N.~Suda, D.~Civin, and V.~Chandra, ``Federated learning
  with non-iid data,'' \emph{arXiv preprint arXiv:1806.00582}, 2018.

\end{thebibliography}

\end{document}